\documentclass[11pt]{article}

\usepackage[T1]{fontenc}
\usepackage[utf8]{inputenc}
\usepackage{authblk}
\usepackage{etex}
\usepackage[english]{babel}
\usepackage[a4paper,margin=1in]{geometry} 
\usepackage{xytree}
\usepackage{amsmath, amssymb} 
\usepackage{amsthm}
\usepackage{enumerate}
\usepackage{enumitem}
\usepackage{appendix}
\usepackage{listings}
\usepackage{color}
\usepackage[font=scriptsize]{subfig} 
\usepackage{fancyhdr}
\usepackage{placeins}
\usepackage{longtable}
\usepackage{dcolumn}
\usepackage{etoolbox}

\makeatletter
\patchcmd {\LT@start}
{\vfil \break}
{\if@nobreak\else \vfil \break \fi}
{\typeout{Patching longtable succeeded!}}
{\typeout{Patching longtable failed!}\ERROR}
\patchcmd {\LT@start}
{\penalty \z@}
{\if@nobreak\else \penalty \z@ \global \@nobreakfalse \fi}
{\typeout{Patching longtable succeeded!}}
{\typeout{Patching longtable failed!}\ERROR}
\makeatother

\makeatletter
\def\@seccntformat#1{\@ifundefined{#1@cntformat}%
   {\csname the#1\endcsname\quad}  
   {\csname #1@cntformat\endcsname}
}
\let\oldappendix\appendix 
\renewcommand\appendix{%
    \oldappendix
    \newcommand{\section@cntformat}{\appendixname~\thesection\quad}
}
\makeatother

\usepackage[hidelinks]{hyperref}
\usepackage{tikz}
\usepackage{tikz-qtree}
\usetikzlibrary{shapes}
\usetikzlibrary{arrows}
\usetikzlibrary{decorations}
\usetikzlibrary{decorations.markings}
\usetikzlibrary{backgrounds}
\usetikzlibrary{positioning}
\usetikzlibrary{calc}
\tikzset{
did/.style={rectangle,draw=none,fill=none,inner sep=0cm},
blu/.style={fill=blue!20,circle},
Blu/.style={fill=blue!50,circle},
red/.style={fill=red!20,circle},
Red/.style={fill=red!50,circle},
white/.style={draw=black,circle,minimum size=.15cm},
backstd/.style={draw=black,fill=white,rounded corners},
backgr/.style={draw=black,fill=green!20,rounded corners},
background rectangle/.style={draw=black,fill=white,rounded corners},
every node/.style=white,
show background rectangle=true,
pre/.style={->,>={latex'},thick}, 
prer/.style={color=red,->,>={latex'},thick},
post/.style={<-,>={latex'},thick}, 
postr/.style={color=red,<-,>={latex'},thick},
Pre/.style={->,>={latex'},ultra thick}, 
Prer/.style={color=red,->,>={latex'},ultra thick},
Post/.style={<-,>={latex'},ultra thick}, 
Postr/.style={color=red,<-,>={latex'},ultra thick},
b/.style={-,>={latex'},thick}, 
B/.style={-,>={latex'},ultra thick}, 
r/.style={color=red,-,>={latex'},thick},
R/.style={color=red,-,>={latex'},ultra thick},
rect/.style={rectangle,minimum height=.01cm}
}

\definecolor{listinggray}{gray}{0.9}
\definecolor{lbcolor}{rgb}{0.9,0.9,0.9}

\theoremstyle{plain}
\newtheorem{theorem}{Theorem}[section] 

\newtheorem{lemma}[theorem]{Lemma}

\newtheorem{corollary}[theorem]{Corollary}

\theoremstyle{definition}
\newtheorem{definition}[theorem]{Definition}

\theoremstyle{remark}
\newtheorem{remark}[theorem]{Remark}

\newcommand{\leql}{\ensuremath\leq_{ql}}
\newcommand{\eql}{\ensuremath=_{ql}}
\newcommand{\pl}{quasilinear}

\newcommand{\sm}{\textsc{Big}}
\newcommand{\bgdomv}{\textsc{BipGDominatedVertex}}
\newcommand{\gdomv}{\textsc{GraphDominatedVertex}}
\newcommand{\hdome}{\textsc{HypergraphDominatedEdge}}
\newcommand{\lalign}{\textsc{LocalStringAlign}}
\newcommand{\doms}{\textsc{SpernerFamily}}
\newcommand{\msimp}{\textsc{MaximalElementsFamily}}
\newcommand{\hasse}{\textsc{SubsetGraph}}
\newcommand{\ortbv}{\textsc{OrthogonalityBinaryVectors}}
\newcommand{\ksat}{\textsc{$k$-Sat}}
\newcommand{\ksatt}{\text{\ksat *}}
\newcommand{\tds}{\textsc{TwoDisjointSets}}
\newcommand{\betv}{\textsc{BetweennessCentralityVertex}}
\newcommand{\bet}{\textsc{BetweennessCentrality}}
\newcommand{\diam}{\textsc{SplitGraphDiameter2Or3}}
\newcommand{\clos}{\textsc{MinimumClosenessCentrality}}
\newcommand{\tcov}{\textsc{TwoCovering}}
\newcommand{\matzero}{\textsc{ZerosMatrixMultiplication}}
\newcommand{\bgdt}{\textsc{Bipartite3DominatingSet}}
\newcommand{\bsdp}{\textsc{BipartiteSubset2DominatingSet}}
\newcommand{\gdt}{\textsc{3DominatingSet}}
\newcommand{\tsum}{\textsc{3Sum}}
\newcommand{\diamg}{\textsc{GraphDiameter2Or3}}
\newcommand{\hyper}{\textsc{HyperbolicityWith2FixedVertices}}

\newcommand{\sms}{Big}
\newcommand{\bgdomvs}{Bipartite \\Graph \\ Dominated \\ Vertex}
\newcommand{\gdomvs}{Graph \\ Dominated \\ Vertex}

\newcommand{\laligns}{Local String \\ Alignment}
\newcommand{\domss}{Sperner \\ Family}
\newcommand{\msimps}{Maximum \\ Simple Family \\ Of Sets}
\newcommand{\hasses}{Subset \\ Graph}
\newcommand{\ortbvs}{Orthogonality of \\ Binary Vectors}

\newcommand{\ksatts}{$k$-Sat*}
\newcommand{\tdss}{Two \\ Disjoint \\ Sets}
\newcommand{\betvs}{Betweenness \\ Centrality of $v$}
\newcommand{\bets}{Betweenness \\ Centrality}
\newcommand{\diams}{Split Graph \\ Diameter 2 or 3}
\newcommand{\closs}{Minimum \\ Closeness \\ Centrality}
\newcommand{\tcovs}{Two \\ Covering}
\newcommand{\matzeros}{Zeros In \\ Matrix \\ Multiplication}
\newcommand{\bgdts}{Bipartite \\ 3-Dominating \\ Set}
\newcommand{\bsdps}{Bipartite \\ Subset \\ 2-Dominating \\ Set}
\newcommand{\gdts}{3-Dominating \\ Set}
\newcommand{\diamgs}{Graph \\ Diam.~2 or 3}
\newcommand{\hypers}{Hyperbolicity \\ with Two \\ Fixed Vertices}

\renewcommand{\O}{\mathcal O}
\newcommand{\otilde}{\tilde{\O}}
\newcommand{\C}{\ensuremath\mathcal{C}}

\newenvironment{prob}[3]
{

\begin{description}[leftmargin=2.3cm,style=sameline,topsep=\baselineskip ,parsep=0pt,itemsep=0pt,partopsep=0pt,font=\normalfont \itshape]
\item[Problem:] \textsc{#1}
\item[Input:] #2
\item[Output:] #3
\end{description}

}

\makeindex

\title{Into the Square \\ On the Complexity of Quadratic-Time Solvable Problems}

\author[1]{Michele Borassi}
\author[2]{Pierluigi Crescenzi}
\author[3,4]{Michel Habib}
\affil[1]{IMT Insitute for Advanced Studies, 55100 Lucca, Italy, \url{michele.borassi@imtlucca.it}}
\affil[2]{Universit\`a di Firenze, Dipartimento di Ingegneria dell'Informazione, 50134 Firenze, Italy, \url{pierluigi.crescenzi@unifi.it}}
\affil[3]{LIAFA, UMR 7089 CNRS \& Universit\'e Paris Diderot, France, \url{habib@liafa.univ-paris-diderot.fr}}

\date{\today}

\begin{document}
\sloppy
\ifdefined \all
\tableofcontents
\fi

\maketitle

\thispagestyle{empty}

\begin{abstract}
In this paper, we will analyze several quadratic-time solvable problems, and we will classify them into two classes: problems that are solvable in \emph{truly subquadratic} time (that is, in time $\O(n^{2-\epsilon})$ for some $\epsilon>0$) and problems that are not, unless the well known Strong Exponential Time Hypothesis (in short, SETH) is false. In particular, we will prove that some quadratic-time solvable problems are indeed easier than expected. We will provide an algorithm that computes the transitive closure of a directed graph in time $\O(mn^{\frac{\omega+1}{4}})$, where $m$ denotes the number of edges in the transitive closure and $\omega$ is the exponent for matrix multiplication. As a side effect of our analysis, we will be able to prove that our algorithm runs in time $\O(n^{\frac{5}{3}})$ if the transitive closure of the graph is sparse. The same time bounds hold if we want to check whether a graph is transitive, by replacing $m$ with the number of edges in the graph itself. As far as we know, this gives us the fastest algorithm for checking whether a sparse graph is transitive. Finally, we will also apply our algorithm to the comparability graph recognition problem (which dates back to 1941): also in this case, we will obtain the first truly subquadratic algorithm. In the second part of the paper we will deal with hardness results. In particular, we will start from an artificial quadratic-time solvable variation of the \ksat\ problem and we will construct a graph of Karp reductions, proving that a truly subquadratic-time algorithm for any of the problems in the graph falsifies SETH. More specifically, the analyzed problems are the following: computing the subset graph, finding dominating sets, computing the betweenness centrality of a vertex, computing the minimum closeness centrality, and computing the hyperbolicity of a pair of vertices. We will also be able to include in our framework three proofs already appeared in the literature, concerning the problems of distinguishing between split graphs of diameter 2 and diameter 3, of solving the local alignment of strings and of finding two orthogonal binary vectors inside a collection.

\end{abstract}

\newpage

\pagenumbering{arabic}

\section{Introduction}

Since the very beginning of theoretical computer science and until recent years, the duality between NP-hard problems and polynomial-time solvable problems has been considered the threshold distinguishing ``easy'' from ``hard'' problems. However, polynomial-time algorithms might not be as efficient as one expects: for instance, in real-world networks with millions or billions of nodes, also quadratic-time algorithms might turn out to be too slow in practice, and a \emph{truly subquadratic} algorithm would be a significant improvement, where an algorithm is said to be truly subquadratic if its time-complexity is $\O(n^{2-\epsilon})$ for some $\epsilon>0$.

\subsection{Subquadratic-Time Results}

In the first part of this paper, we will analyze two well-known problems (that is, checking whether a graph is transitive and recognizing comparability graphs) and we will show that they are solvable in truly subquadratic time. For what concerns the transitivity problem, our main contribution is a new analysis of an old algorithm that finds the transitive closure of a graph \cite{Goralcikova1979}. This analysis will lead us to a simple modification of the algorithm itself, that will provide an $\O\left(mn^{\frac{\omega+1}{4}}\right)$ algorithm, where $m$ is the number of edges in the transitive closure and $\omega$ is the exponent for matrix multiplication, whose current value is $2.3727$ \cite{Williams2012}. As a consequence, we will be able to check if a graph is transitive in truly subquadratic time: as far as we know, no truly subquadratic algorithm was previously known for this problem, although many papers have been published on the computation of the transitive closure \cite{Chen2003,Blelloch2008,Lacki2011}. As a side effect of our analysis, we will be able to prove that our algorithm runs in time $\O(n^{\frac{5}{3}})$ in graphs whose transitive closure is sparse: as far as we know, this gives us the fastest algorithm for checking whether a sparse graph is transitive. More importantly, combining this result with the results in \cite{Habib2000}, we will be able to provide a truly subquadratic algorithm for recognizing comparability graphs, a widely studied graph class (for more information on comparability graphs, we refer to~\cite{Brandstadt1999}). As far as we know, the existence of such an algorithm was also not known before our result, even if this class of graphs is quite old (the oldest mention we were able to find dates back to 1941~\cite{Dushnik1941}). 


\subsection{Hardness Results}

Following the main ideas behind the theory of NP-completeness, and not being able to show that a specific polynomial-time solvable problem might or might not admit a faster algorithm, researchers have recently started to prove that the existence of such an algorithm would imply faster solution for many other problems. As an example, a huge amount of work started from the analysis of the \tsum\ problem, which consists of, given three sets $A$, $B$ and $C$ of integers, deciding whether there exists $a \in A, b \in B$, and $c \in C$ such that $a+b+c=0$. This problem has been widely studied and, as far as we know, the best algorithm has been provided in \cite{Baran2007}: this algorithm is subquadratic, but not truly subquadratic. The \tsum\ problem has then become a starting point for proving the ``hardness'' of many other problems, especially in algebraic geometry (for example, we refer the interested reader to \cite{King2004}). Notice that all these results do not deal with the notion of completeness, but they simply prove that ``a problem is harder than another'', relying on the fact that the easiest problem has been studied for years and no efficient algorithm has been found.

A more recent develop of this field is based on the Strong Exponential Time Hypothesis (SETH), used as a tool to prove the hardness of polynomial-time solvable problems. This hypothesis, stated in \cite{Impagliazzo2001}, says that there is no algorithm for solving the \ksat\ problem in time $\O((2-\epsilon)^n)$, where $\epsilon>0$ does not depend on $k$. Successively, researchers have started to use it in order to prove hardness results (see for example \cite{Williams2010,Patrascu2010}, where the authors address the hardness of many problems, like the all-pairs-shortest-paths, finding triangles in a graph, 2-\textsc{Sat}, $k$-dominating set, and some generalization of matrix multiplication). Starting from these works, many other hardness results based on SETH have been published, and many of them deal with dynamic problems (see, for instance, \cite{Abboud2014a}). As an example, it is worth mentioning the diameter computation, that is, given a graph, finding the maximum distance between two vertices. Despite numerous papers on the topic, no truly subquadratic algorithm has been found so far (for more details on the diameter computation, we refer to \cite{Takes2011,Takes2013,Crescenzi2013,Borassi2014}). However, a reason for this behavior was found in \cite{Roditty2013}, where it is proved that a truly subquadratic algorithm to distinguish graphs of diameter 2 and 3 would falsify SETH (see also \cite{Roditty2014}). This result is meaningful both because it gives a lower bound on the complexity of the diameter problem, and because it sheds some light on SETH, which is becoming more and more central in modern computer science. Other similar result are the hardness of the local sequence alignment problem, proved in \cite{Abboud2014}, and the hardness of finding two orthogonal binary vectors in a collection \cite{Williams2014}.

On the ground of this approach, in the second part of this paper we will show that several well-known quadratic-time solvable problems are not solvable in truly subquadratic time, unless SETH is false. As a first step, we will define the problem \ksatt\ (obtained through an ``artificial'' modification of the input of the \ksat\ problem), which cannot be solved in truly subquadratic time, unless SETH is false (note that the use of the \ksatt\ problem has already been implicitly suggested in \cite{Williams2010}). We will then design several Karp-reductions that preserve truly subquadratic-time resolvability, and we will use these reductions in order to prove the hardness of several other problems (note that all problems considered, apart from the $3$-dominating set, are actually solvable in quadratic time, so that the lower and upper bounds coincide, apart from logarithmic factors). More specifically, the analyzed problems are the following.

\begin{description}
\item[\hasse:] given a collection $\C$ of subsets of a given ground set $X$, find the subset graph of $\C$. The subset graph is defined as a graph whose vertices are the elements of $\C$, and containing an edge $(C,C')$ if $C \subseteq C'$. For this problem, the first subquadratic algorithm was proposed in \cite{Yellin1993}, and in \cite{Pritchard1999,Elmasry2010} matching lower bounds are proved. However, these lower bounds are based on the number of edges in the subset graph, which might be quadratic with respect to the input size, apart from logarithmic factors. Our results show that the complexity of computing the subset graph is not due to the output size only, but it is intrinsic: in particular, we will prove that even deciding whether the subset graph has no edge is hard. This excludes the existence of a truly subquadratic algorithm to check if a solution is correct, or a truly subquadratic algorithm algorithm for instances where the output is sparse.

\item[\betv:] the betweenness centrality of a vertex in a graph is a widely used graph parameter related to community structures, defined in \cite{Freeman1977} (for more information we refer to the book \cite{Newman2010} and the references therein). Despite numerous attempts like \cite{Brandes2001,Bader2007,Edmonds2010}, there exist no truly subquadratic algorithm computing the betweenness centrality, even of a single vertex. Moreover, in \cite{Bader2007}, it is said that finding better results for approximating the betweenness centrality of all vertices (\bet) is a ``challenging open problem''. Our analysis does not only prove that computing the betweenness centrality of all vertices in subquadratic time is against SETH, but it also presents the same result for computing the betweenness of a single vertex.

\item[\clos:] another fundamental parameter in graph analysis is closeness centrality, defined for the first time in 1950 \cite{Bavelas1950} and recently reconsidered when analyzing real-world networks (for the interested reader, we refer to \cite{Latora2007} and the references therein). This problem has also raised algorithmic interest, and the most recent result is a very fast algorithm to approximate the closeness centrality of all vertices \cite{Cohen2014}. In this paper, we will prove for the first time the hardness of finding the ``least central'' vertex with respect to this measure. Simple consequences of this results are the hardness of computing in truly subquadratic time the closeness centrality of all vertices, or of extracting a ``small enough'' set containing all peripheral vertices.

\item[\hyper:] the Gromov hyperbolicity of a graph \cite{Gromov1987} is another parameter that recently got the attention of researchers in the field of network analysis. This parameter has relations with the chordality of a graph \cite{Wu2011}, and with diameter and radius computation \cite{Chepoi2008,Crescenzi2013}. In \cite{Cohen2012}, it is provided an algorithm with time-complexity $\O(n^4)$, but much more efficient in practice, while in \cite{Fournier2012} an algorithm with time-complexity $\O(n^{3.69})$ is described. The latter paper also provides an algorithm with complexity $\O(n^{2.69})$ to compute the hyperbolicity when a vertex is fixed. One might think that if two vertices are fixed, then the bound becomes subquadratic: we will prove that this is not true unless SETH is false. This is a significant result on its own right, but it can also be used in the design of algorithms that compute the hyperbolicity of the whole graph.

\item[\gdt:] the last ``classical'' problem analyzed is finding dominating sets in a graph, which is one of the 21 Karp's NP-complete problems \cite{Karp1972}. Researchers have mainly tried to efficiently determine a $k$-dominating set when $k$ is fixed \cite{Eisenbrand2004}. Also negative results were published on the minimum dominating set problem: in \cite{Patrascu2010} it is proved that finding a dominating set of size $k$ in $\O(n^{k-\epsilon})$ would falsify SETH. Our paper proves that a subquadratic-time algorithm for the $3$-dominating set problem would falsify SETH. This result is complementary to the result of \cite{Patrascu2010}: their result works better on dense graphs, while our result works better on sparse graphs. We will also prove that, given a bipartite graph $G=(V,W,E)$, finding a pair of vertices in $V$ that dominates $W$ is hard.

\end{description}

We will also be able to include in our framework three proofs already appeared in the literature, concerning the problems of distinguishing between split graphs of diameter 2 and diameter 3 \cite{Roditty2013}, of solving the local alignment of sequences \cite{Abboud2014}, and of finding orthogonal vectors in a collection \cite{Erickson1996}. 

In order to prove our reductions, we have also analyzed several other problems, that play the role of ``intermediate steps''. These intermediate steps are very natural problems, and they can be considered as significant results on their own right. For instance, many of these problems deal with a collection $\mathcal{C}$ of subsets of a given ground set $X$, and try to find particular pairs in the collection (\tds, \doms, \tcov). These problems are very natural, but we have not been able to find any paper on the topic, apart from the \doms\ case, that is, finding whether there exists two sets $C \subseteq C'$ in the collection $\C$. Such a family, also named clutter, has been analyzed in many cases (see for example \cite{Engel1997,Bey2002}), but no recognition algorithm has been provided, yet. For each of these set problems, this paper deals with two variations: in the first one, the size of the collection is exponential with respect to the size of $X$, while in the other one any size is accepted (so the problem becomes harder). Other problems have been obtained by ``rephrasing'' these problems in terms of graphs, for example finding dominated vertices (\bgdomv\ and \gdomv). We believe that these results are interesting in their own right, but they also can be used as intermediate points for new reductions, without restarting from SETH for any new hardness proof.

\subsection{Structure of the Paper}

In Section \ref{sec:easy} we will discuss the two problems that, although difficult at first glance, actually can be solved in truly subquadratic time. In Section \ref{sec:hard} we will provide the framework in which all hardness results are given and we will state and prove all reductions. Section \ref{sec:conc} concludes the paper and provides some open problems. 

\section{Transitive Closure and Comparability Graph Test}\label{sec:easy}

In this section, we will prove that the complexity of recognizing transitive directed graphs is truly subquadratic. As a result, we will also be able to provide a truly subquadratic algorithm to recognize comparability graphs.

Let us start by assuming without loss of generality that the input graph is acyclic, since the transitive closure of a connected component is a clique, and it is possible to find connected components in linear time.

The new algorithm on acyclic graphs relies heavily on a new analysis of the old algorithm published in \cite{Goralcikova1979} (the pseudo-code is given in Algorithm \ref{alg:trans}). This analysis is provided by the following theorem.

\begin{theorem} \label{thm:transbase}
Let $\alpha$ be any real value in $[0,1]$. Algorithm \ref{alg:trans} finds the transitive closure of the input graph $G$ in time $\O(\frac{m^2}{n^{2\alpha-1}}+mn^\alpha)$, where $m$ is the number of edges in the transitive closure.
\end{theorem}

Before proving the theorem, we will apply it in the design of truly subquadratic algorithms, both for checking transitivity closure and for comparability graph recognition. 
Note that the algorithm for sparse graphs is practical, since it is purely combinatoric and it does not rely on matrix multiplication like the one in \cite{Blelloch2008}.

\begin{corollary}
If the transitive closure of a graph is sparse, then it can be computed in time $\O(n^\frac{5}{3})$.
\end{corollary}
\begin{proof}
Apply the previous theorem with $\alpha=\frac{2}{3}$.
\end{proof}

In order to obtain a truly subquadratic algorithm for any graph, we need to pair the previous algorithm with matrix multiplication: as a consequence, the algorithm in the following corollary might have huge constants hidden inside the $\O$ notation, differently from the previous one. 

\begin{corollary}
It is possible to compute the transitive closure of any graph in time $\O(mn^{\frac{\omega+1}{4}})$, where $m$ is the number of edges in the transitive closure.
\end{corollary}
\begin{proof}
In \cite{Kozen1992}, Algorithm 5.2, it is provided an algorithm that computes the transitive closure in time $\O(n^\omega)$. We will use this algorithm if $m\geq n^{\frac{3\omega-1}{4}}$, and the time complexity is $\O(n^\omega) \leq \left(mn^{\frac{\omega+1}{4}}\right)$. Otherwise, that is, if $m<n^{\frac{3\omega-1}{4}}$, we will use Algorithm \ref{alg:trans}: applying Theorem \ref{thm:transbase} with $\alpha=\frac{\omega+1}{4}$, we obtain a time complexity of $\O\left(\frac{m^2}{n^{2\alpha-1}}+mn^\alpha\right)=\O\left(m\frac{n^{\frac{3\omega-1}{4}}}{n^{\frac{\omega-1}{2}}}+mn^{\frac{\omega+1}{4}}\right)=\O\left(mn^{\frac{\omega+1}{4}}\right)$.

However, the value of $m$ is not known in advance: in any case, it is enough to start by applying Algorithm \ref{alg:trans} until the number of edges reaches $n^\alpha$, and continue with matrix multiplication if this is the case.
\end{proof}

The following corollary applies the previous results to transitive graph recognition.

\begin{corollary}
It is possible to check if a graph is transitive in time $\O(\frac{m^2}{n^{2\alpha-1}}+mn^\alpha)$, and $\O(n^\frac{5}{3})$ if the graph is sparse.
\end{corollary}
\begin{proof}
It is easy to see that if the previous algorithms are stopped as soon as a single edge is added, the running time bounds depend on the input size and not on the output size.
\end{proof}

Finally, next corollary will apply the previous algorithms to the comparability graph recognition problem, providing again the first truly subquadratic algorithm for this task.

\begin{corollary}
It is possible to check if a graph is a comparability graph in time $\O\left(mn^{\frac{\omega+1}{4}}\right)$. If the graph is sparse, the running time is $\O\left(n^{\frac{5}{3}}\right)$
\end{corollary}
\begin{proof}
It is known that, if a graph $G$ is a comparability graph, then it is possible to compute a transitive orientation of $G$ in linear time \cite{McConnellS1999} (for a simpler algorithm running in $\O(n+m\log n)$ see \cite{Habib2000}). With input a graph $G$, our comparability test runs the transitive orientation algorithm: let $H$ be the output of this algorithm. Then, we use the previous transitive closure algorithm (either the one for sparse graphs or the general one), in order to decide if $H$ is transitive. If it is, $G$ is clearly a comparability graph. If $H$ is not transitive, it means that the graph $G$ is not a comparability graph, since otherwise the transitive orientation algorithm would have provided a transitive orientation.

The running time of comparability graph recognition coincides with the running time of the transitive closure algorithm, since all other steps are faster. This proves the theorem.
\end{proof}

The rest of this section is devoted to the proof of Theorem \ref{thm:transbase}, which analyzes the complexity of Algorithm \ref{alg:trans}.

\begin{small}
\begin{lstlisting}[caption=Computing the transitive closure of a graph.,label={alg:trans}]
TransClosure($G=(V,N)$) {
	find a topological ordering $(\sigma_1,\dots,\sigma_n)$ of $V$;
	for (int i = n; i > 0; i--) {
		$N'(\sigma_i)$ = $N(\sigma_i)$;
		for w in $N_i(v)$ {
			$N'(\sigma_{i})$ = $N'(\sigma_i)$ $\cup$ $N'(w)$
		}
	}
	return ($V$, $N'$)
}
\end{lstlisting}
\end{small}

We will divide the proof into a sequence of lemmas. For completeness, we also include the correctness proof (Lemmas \ref{lem:transcond} and \ref{lem:cor}), already provided in \cite{Goralcikova1979}. First of all, let us restate the definition of transitive closed graph, so that it is easily checkable.

\begin{lemma}\label{lem:transcond}
A directed acyclic graph is transitive if and only if for each edge $(v,w)$, $N(v) \supseteq N(w)$.
\end{lemma} 
\begin{proof}
If the graph is transitive, and $x \in N(w)$, then there is the path $(v,w,x)$, and by transitivity $x \in N(v)$. For the other direction, let us assume there is a path from $v$ to $w$ and let us prove that $v$ is linked to $w$ by induction on the length $k$ of the path (the base case holds by hypothesis). For induction step, let $v_1$ be the first vertex of the path after $v$: $v_1$ is linked to $w$ by inductive hypothesis, and we conclude because $N(v) \supseteq N(v_1)$.
\end{proof}


%

\begin{lemma}\label{lem:cor}
The algorithm is correct, that is, the output graph is the transitive closure of the input graph.
\end{lemma}
\begin{proof}
We will prove by backward induction that after vertex $\sigma_i$ is analyzed in the \texttt{for} cycle at line 4, the condition in Lemma \ref{lem:transcond} is verified for each pair $(\sigma_j,\sigma_k)$ with $j \geq i$, for neighbors $N'$.

The base step is trivial, since $\sigma_n$ has no outgoing edge. For induction step, if $j>i$ there is nothing to prove, since the only modified neighbor in the \texttt{for} cycle is $N'(\sigma_i)$. For $j=i$, we observe that if there is a path from $(\sigma_i,\sigma_k,\dots,\sigma_j)$, then $\sigma_j \in N'(\sigma_k)$ by induction hypothesis, and $\sigma_j \in N'(\sigma_i)$ because of line 7 of the algorithm.

This means that the output graph is transitive. Moreover, the step in line 7 of the algorithm does not modify reachability, so the output graph is not bigger than the transitive closure.
\end{proof}

\begin{proof}[Proof of Theorem \ref{thm:transbase}]
We first observe that the steps in lines 2 and 5 can be performed in linear time, so they have no effect on the running time of the algorithm. The ``hardest'' step is line 7.

In order to estimate the time needed to perform line 7, we define two different sets and we compute separately the running time for each of these sets:

$$X_\alpha:=\{i:|N'(\sigma_i)|\leq n^\alpha\}$$
$$Y_\alpha:=X_\alpha^C=\{i:|N'(\sigma_i)|>n^\alpha\}.$$

The running time needed to perform the instruction in line 7 for an edge $(v,w)$ is $|N(w)|$, by using Fibonacci heaps \cite{Fredman1987}. The time needed to perform this operation for all edges $(v,w)$ with $w \in X_\alpha$ is:

$$\sum_{(v,w) \in E \wedge w \in X_\alpha} |N(w)| \leq mn^\alpha.$$

In order to estimate the time needed to perform the check for edges $(v,w)$ with $w \in Y_\alpha$, we observe that if $w \in Y_\alpha$, then $v \in Y_\alpha$ because $N'(v) \supseteq N'(w)$ and that $n^\alpha|Y_\alpha|\leq m$. Then, the time needed to check all these edges is at most:

$$\sum_{(v,w) \in E \wedge w \in Y_\alpha} |N(w)| \leq n|Y_\alpha|^2 \leq \frac{m^2}{n^{2\alpha - 1}}.$$

Then, for each $\alpha$, the total running time of the algorithm is at most $\O\left(\frac{m^2}{n^{2\alpha-1}}+mn^\alpha\right)$.

\end{proof}

\section{Hard Problems}\label{sec:hard}

In the previous section, we have provided truly subquadratic algorithms for two important problems. The goal of this section is the converse: proving that it is impossible to find truly subquadratic algorithms for some problems, unless SETH is false. We will provide the context under which all these reductions fall in, and in the last part we will prove them.

The starting point of our reductions is an ``artificial'' variation of \ksat\ which is quadratic-time solvable, but not solvable in $\O(n^{2-\epsilon})$ unless SETH is false.

\begin{prob}
{\ksatt.}{two sets of variables $\{x_i\}$, $\{y_j\}$ of the same size, a set $C$ of clauses over these variables, such that each clause has at most size $k$, the set of possible evaluations of $x_i$ and the set of possible evaluations of $\{y_j\}$.}{\textbf{True} if there is an evaluation of all variables that satisfies all clauses, \textbf{False} otherwise.}
\end{prob}

It should be noticed that this problem differs from the classic one only by the input size. This way, a quadratic-time algorithm exists (trying all possible evaluations). However, an algorithm running in $\O(n^{2-\epsilon})$ with $\epsilon$ not depending on $k$ would imply an algorithm solving \ksat\ in $\O(2^{\frac{n}{2}(2-\epsilon)}) = \O((2^{\frac{2-\epsilon}{2}})^n)$, and this is against SETH.

After defining the ``starting'' problem, we need to define reductions. We will use \pl\ reductions from one problem to another.

\begin{definition}
A \pl\ Karp reduction from a problem $\mathcal{P}$ to problem $\mathcal{Q}$ is a function $\Phi$ from instances of $\mathcal{P}$ to instances of $\mathcal{Q}$ verifying for every instance $I$ of $\mathcal{P}$:
\begin{itemize}
\item $\Phi(I)$ can be computed in time $\otilde(s(I))$, where $s(i)$ is the size of input $I$;\footnote{By $\otilde(f(n))$ we mean $\O(f(n)\log^k n)$ for some fixed $k$.}
\item $I$ and $\Phi(I)$ have the same output.
\end{itemize}
In general, if the output is not boolean, we will require a linear-time computable function that transforms the output of $\Phi(I)$ into the output of $I$. If $\mathcal{P}$ is reducible to $\mathcal{Q}$, we will say $\mathcal{P} \leql \mathcal{Q}$.
\end{definition}

\begin{figure}[h!t]

\tikzset{>=stealth'}
\begin{tikzpicture}[every node/.style={draw, rectangle, align=center,font=\footnotesize},label/.style={did, inner sep=1pt,midway,above,sloped,font=\tiny}]
\node (ksatt) at (1,0) {\ksatts};

\node (sdoms) at (3.8,2) {\sms\ \domss};
\node (stds) at (3.8,0) {\sms\ \tdss};
\node (stcov) at (3.8,-2.5) {\sms\ \tcovs};

\node (doms) at (7.6,3.4) {\domss};
\node (tds) at (7.6,2) {\tdss};
\node (betv) at (7.6,0.7) {\betvs};
\node (diam) at (7.6,-.9) {\diams};
\node (matzero) at (7.6,-2.75) {\matzeros};
\node (tcov) at (7.6,-4.5) {\tcovs};
\node (lalign) at (7.6,-6.5) {\laligns};

\node (bgdomv) at (11.6,5) {\bgdomvs};
\node (msimp) at (11.6,3.4) {\msimps};
\node (ortbv) at (11.6,2) {\ortbvs};
\node (clos) at (11.6,-0.2) {\closs};
\node (diamg) at (11.6,-1.4) {\diamgs};
\node (bgdt) at (11.6,-3) {\bgdts};
\node (bsdp) at (11.6,-5) {\bsdps};

\node (gdomv) at (15,5) {\gdomvs};
\node (hasse) at (15,3.4) {\hasses};
\node (hyper) at (15,-1.4) {\hypers};
\node (bet) at (15,0.7) {\bets};
\node (gdt) at (15,-3) {\gdts};

\draw[<-, thick] (ksatt) -- (stds) node[label] {Thm.~\ref{thm:stds}};

\draw[<->, thick] (sdoms) -- (stds) node[label, left, sloped=false] {Thm.~\ref{thm:sdoms}-\ref{thm:sdomsi}};
\draw[<->, thick] (stds) -- (stcov) node[label, left, sloped=false] {Rem.~\ref{rem:tcov}};

\draw[->, thick] (doms) -- (sdoms) node[label] {Rem.~\ref{rem:sm}};
\draw[->, thick] (tds) -- (stds) node[label] {Rem.~\ref{rem:sm}};
\draw[->, thick] (diam) -- (stds) node[label] {Thm.~\ref{thm:diam}};
\draw[->, thick] (matzero) -- (stds) node[label] {Thm.~\ref{thm:matzero}};
\draw[->, thick] (tcov) -- (stcov) node[label] {Rem.~\ref{rem:sm}};
\draw[->, thick] (lalign) -- (stcov) node[label] {Thm.~\ref{thm:lalign}};
\draw[->, thick] (betv.195) -- (stds) node[label] {Thm.~\ref{thm:betv}};
\draw[->, thick] (clos) -- (stds) node[label,very near start,above] {Thm.~\ref{thm:clos}};

\draw[<->, thick] (ortbv) -- (tds) node[label] {Rem.~\ref{rem:ortbv}};

\draw[<->, thick] (bgdomv) -- (doms) node[label] {Rem.~\ref{rem:bgdomv}-\ref{rem:bgdomvi}};
\draw[->, thick] (msimp) -- (doms) node[label] {Rem.~\ref{rem:msimp}};
\draw[->, thick] (bgdt) -- (tcov) node[label] {Thm.~\ref{thm:bgdt}};
\draw[->, thick] (bsdp) -- (tcov) node[label] {Thm.~\ref{thm:bgdt}};
\draw[->, thick] (diamg) -- (diam) node[label] {Rem.~\ref{rem:sm}};

\draw[<->, thick] (gdomv) -- (bgdomv) node[label] {Rem.~\ref{rem:sm}-\ref{rem:gdomvi}};
\draw[->, thick] (hasse) -- (msimp) node[label] {Rem.~\ref{rem:hasse}};
\draw[->, thick] (bet) -- (betv) node[label] {Rem.~\ref{rem:bet}};
\draw[->, thick] (gdt) -- (bgdt) node[label] {Rem.~\ref{rem:sm}};
\draw[->, thick] (hyper) -- (diamg) node[label] {Thm.~\ref{thm:hyper}};
\end{tikzpicture}
\caption{A scheme of the reductions provided by this article.} \label{fig:scheme}
\end{figure}
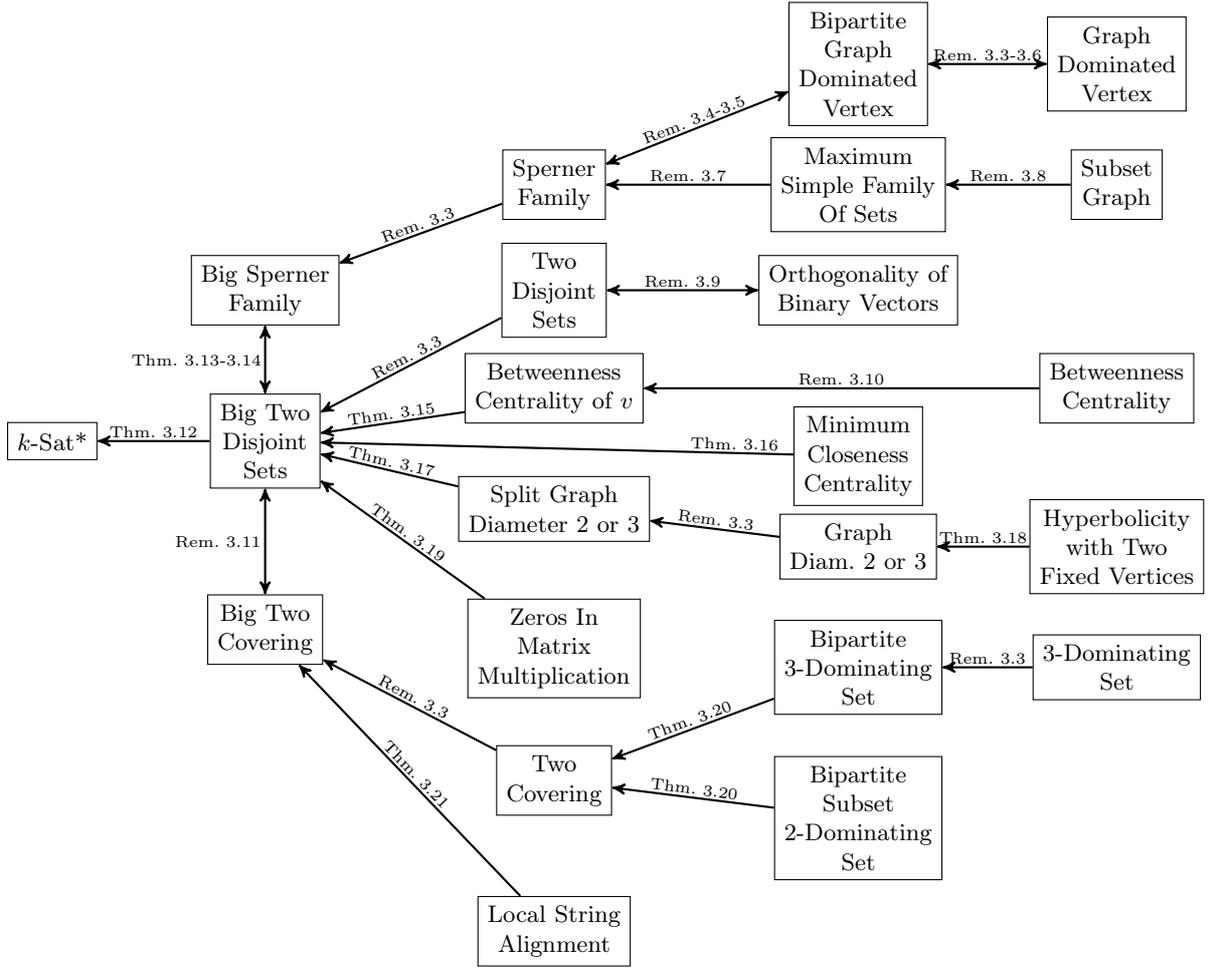

\begin{remark}
If $\mathcal{P} \leql \mathcal{Q}$ and there is an algorithm solving $\mathcal{Q}$ in time $\otilde(n^{2-\epsilon})$ for some $\epsilon$, then $\mathcal{P}$ can be solved in time $\otilde(n^{2-\epsilon})$. 
\end{remark}

The reductions that we will provide are summarized in Figure \ref{fig:scheme} (the definition of all problems is given in Appendix \ref{app:prob}). The previous remark implies that an $\O(n^{2-\epsilon})$ algorithm for any of these problems would falsify SETH. All proofs of those reductions will be provided by the next subsection.

\subsection{Proof of Reductions}

In this subsection, we will prove all the reductions in Figure \ref{fig:scheme}. We have divided these reductions in remarks and theorem, depending on how intricate those constructions are. We will start by providing the proof of remarks, in the order in which they appear in Figure \ref{fig:scheme}.

\begin{remark} \label{rem:sm}
For each problem, $\mathcal{P}$, \sm$\mathcal{P} \leql \mathcal{P}$, since the instances of \sm$\mathcal{P}$ are a subset of the instances of $\mathcal{P}$ and the required output is the same (the function $\Phi$ in the definition of reductions is the identity). The same argument proves that $\bgdomv \leql \gdomv$, that $\bgdt \leql \gdt$ and that $\diam \leql \diamg$.
\end{remark}

\begin{remark}\label{rem:bgdomv}
$\hdome \leql \bgdomv$: each hypergraph $(V,E)$ can be transformed to a biparted graph $(V,W,A)$ where $W=E$ and $A(v,e)$ holds if $v \in e$ (hypergraph edges are seen as sets of vertices). It is clear that an edge $e$ dominates $e'$ in the hypergraph if and only if the corresponding element in $W$ does. The only problem is that two vertices $v,v' \in V$ might also dominate each other: to this purpose, we add to $W$ another copy of $V$ named $V'$ and we add to $E$ all edges $(v,v)$ from $V$ to $V'$.
\end{remark}

\begin{remark}\label{rem:bgdomvi}
$\bgdomv \leql \hdome$: starting from the bipartite graph $(V,V',A)$, create two hypergraphs $(V,E)$ and $(V',E')$. The edges are defined as follows: for each $v \in V$ we add $N(v)$ to $E'$ (since $N(v)$ is contained in $V'$), and likewise for each $v' \in V'$ we add $N(v')$ to $E$. This way, if an edge $v \in V$ dominates an edge $w \in V$, the corresponding edges in $V'$ in the second hypergraph dominate each other. Conversely, if the dominated vertex is in $V'$, there is a dominated edge in the first hypergraph. In order to conclude the reduction, it is enough to consider the disjoint union of $(V,E)$ and $(V',E')$.
\end{remark}

\begin{remark}\label{rem:gdomvi}
$\gdomv \leql \bgdomv$: given a graph $G=(V,E)$, let us construct a bipartite graph $G'=(V,V',A)$ where $V$ and $V'$ are copies of $V$ and $A(v,v')$ holds if $v \in V$, $v' \in V'$ and $E(v,v')$ holds in $G$. It is clear that this construction preserves domination.
\end{remark}


\begin{remark}\label{rem:msimp}
$\doms \leql \msimp$, since a family $\C$ of sets has a dominated set if and only if it is not simple, if and only if the maximum simple family of sets is $\C$.
\end{remark}

\begin{remark}\label{rem:hasse}
$\msimp \leql \hasse$, since the maximum simple family of sets is the set of sources of the subset graph.
\end{remark}

\begin{remark}\label{rem:ortbv}
$\ortbv \eql \tds$: it is enough to choose an ordering on $X$ and to code a set $C \subseteq X$ as a vector of length $|X|$ having $1$ at place $x$ if $x \in C$, $0$ otherwise. Two such vectors are orthogonal if and only if the set of ones in these two vectors are disjoint (we are assuming that the ring over which the vector are considered has characteristic 0, that is, in that ring, $1+1+\dots+1$ is never $0$). We are also assuming a ``clever'' data structure to store binary vectors, that is, no space is needed to memorize zeros.
\end{remark}

\begin{remark}\label{rem:bet}
$\betv \leql \bet$, because computing the betweenness centrality of all vertices is of course harder than computing the betweenness centrality of a single vertex.
\end{remark}

\begin{remark}\label{rem:tcov}
$\sm\tds \eql \sm\tcov$, because $C,C' \in \C$ are disjoint if and only if $C^C,C'^C$ cover $X$: it is enough to define $\Phi(X,\C)=(X,\bar{\C})$, where $\bar{\C} = \{C^C:C \in \C\}$. Note that the input size in the reduction can increase only by a logarithmic factor, since the size of $X$ is $\O(\log^k(|\C|))$.
\end{remark}

We may now turn to the proof of the theorems. Again, the theorems are sorted according to Figure \ref{fig:scheme}.

\begin{theorem}\label{thm:stds}
For each $k$, $\ksatt \leql \sm \tds$.
\end{theorem}
\begin{proof}
Let $X_1$ be the set of the possible evaluations of $\{x_i\}$, $X_2$ the set of possible evaluations of $\{y_j\}$, $C$ the set of clauses of an instance $I$ of \ksatt. We define $\Phi(I)=(X,\C)$, where $X=C \cup \{x_1,x_2\}$, and $\C$ is the collection made by sets of clauses not satisfied by an assignment in $X_1$ or $X_2$ (and $x_1,x_2$ are used to distinguish between assignments in $X_1$ and $X_2$).

More formally, $\C=\C_1 \cup \C_2$, $\C_1:=\{\{x_1\}\cup \{c \in C: x \not\models c\}: x \in X_1\}$ and similarly $\C_2:=\{\{x_2\}\cup \{c \in C: x \not\models c\}: x \in X_2\}$. This way, the size of $\Phi(I)$ is linear in the size of $I$. Moreover, it is possible to compute $\Phi$ by analyzing all sets of evaluation of variables one by one, and for each of them check which clauses are verified. For each evaluation in $X_1$ or $X_2$, the checking time is proportional to the number of clauses, which is at most $\O(\log^k(n))$, where $n$ is the input size of \ksatt.

It remains to prove that the output of the problem is preserved: we will prove that there is a bijection between the pairs of disjoint sets and the satisfying assignments of the formula of \ksat*.

In particular, two sets that are both in $\C_1$ or both in $\C_2$ cannot be disjoint because of $x_1$ and $x_2$. As a consequence, two disjoint sets correspond to an evaluation of all variables: the evaluation satisfies $\phi$ if and only if for each clause there is a variable contained in the evaluation if and only if there is no clause contained in both sets.

\end{proof}

\begin{theorem}\label{thm:sdoms}
$\sm \tds \leql \sm \doms$.
\end{theorem}
\begin{proof}
Consider an instance $I=(X,\mathcal{C})$ of $\textsc{\sm TwoDisjointSets}$. First of all, we define $\Phi'(I)=(X,\mathcal{C'})$, where $\C' = \C \cup \bar{\C}$, and $\bar{\C}:=\{C^C:C \in \C\}$ (which is not the correct definition, but we will see how to adapt it).

If we find two sets $C \in \C, C' \in \bar{\C}$ such that $C \subseteq C'$, we know that $C$ and $C'^C$ are disjoint and in $\C$, so we have found a solution. However, we might also find two sets $C \subseteq C', C' \in \C$, $C \subseteq C'$ with $C \in \bar{\C}$ and $C' \in \C$, $C \subseteq C', C' \in \bar{\C}$. The remaining part of the proof slightly modifies the set $X$ and the collection $\C'$ in order to avoid such cases.

The first problem can be solved by defining $k:=\lceil \log_2(|\C|)\rceil$, and adding two sets $Y=\{y_1,\dots,y_k\}$ and $Z=\{z_1,\dots,z_k\}$ to $X$. In particular, we add $Y$ and $Z$ to each set in $\bar{\C}$ and we add to each element $C \in {\C}$ some $y_i$ and some $z_j$, so that no element of $\C$ can dominate another element in $\C$ (for example, we may associate each set $C$ with a unique binary number with $k$ bits, and code this number using $y_i$ as zeros and $z_j$ as ones). In order to solve the second problem, it is enough to make the same construction adding new sets $Y'$ and $Z'$ of logarithmic size, and use them to uniquely code any element in $\bar{\C}$. None of the elements in $Y'$ and $Z'$ is added to elements in $\C$, and this also solves the third problem.
\end{proof}

\begin{theorem}\label{thm:sdomsi}
$\sm \doms \leql \sm \tds$.
\end{theorem}
\begin{proof}
Let us consider an instance $I=(X,\C)$ of $\textsc{\sm DominatedSet}$ and let us define $\Phi(I)=(X \cup{x_1,x_2},\mathcal{C'})$, where $\C':=\C_1 \cup \C_2$, $\C_1:=\{C \cup {x_1}:C \in \C\}$ and $\C_2:= \{C \cup {x_2}:C \in \bar{\C}\}$.

If $C_1$ and $C_2$ are disjoint sets in $\C'$, one of them must be in $\C_1$ and the other in $\C_2$ (because of $x_1$ and $x_2$). As a consequence, there are two disjoint sets $C_1 \in \C_1,C_2 \in \C_2$ in $\Phi(I)$ if and only if $C_1 \cap X \in \C$, $C_2^C \cap X \in \C$ are disjoint if and only if $C_1 \cap X \subseteq C_2 \cap X$ (and this means that $\C$ contains a set dominating another).
\end{proof}

\begin{theorem}\label{thm:betv}
$\sm \tds \leql \betv$.
\end{theorem}
\begin{proof}
Let us consider an instance $I=(X,\C)$ of $\sm \tds$, and let us construct a graph $G=(V,E)$ as follows:
\begin{itemize}
\item $V:=\{y\} \cup \{x\} \cup \C_x \cup X \cup \C_y$, where $X$ is the ground set of $I$ and $\C_x,\C_y$ are two identical copies of $\C$ (in the graph, $y,x,\C_x,X,\C_y$ will somehow resemble a cycle).
\item all pairs of vertices in $V'$ are connected;
\item vertex $x$ is connected to $y$ and to each vertex in $\C_x$;
\item vertex $y$ is connected to $x$ and to each vertex in $\C_y$;
\item connections between $\C_x$ and $X$ and connections between $\C_y$ and $X$ are made according to the $\in$-relation.
\end{itemize}
The input vertex of our problem is $x$.

We observe that the graph is a big ``cycle'' made by five ``parts'': $y,x,\C_x,X,\C_y$. Moreover, it is possible to move from one part to any vertex in the next part in one step (except if we start or arrive in $X$, and in that case we need at most two steps). As a consequence, no shortest path can be longer than $3$.

This proves that any shortest path in the sum passing through $x$ must be of one of the following forms:
\begin{itemize}
\item a path from $y$ to $\C_x$;
\item a path from $\C_x$ to $\C_y$;
\item a path from $y$ to $X$
\end{itemize}

We note that the third case never occurs, because there always exists a path from $y$ to any vertex in $X$ of length $2$. The first case occurs for each vertex in $V_x$, and no other shortest path exists from $y$ to vertices in $\C_x$: these vertices contribute to the sum by $|\C|$. Finally, the second case occurs only if and only if there is a pair of vertices in $\C_y$ and $\C_x$ having no path of length 2, that is, two disjoint sets in $\C$. 

This proves that the betweenness of $x$ is bigger than $|\C|$ if and only if there are two disjoint sets in $\C$.
\end{proof}

\begin{theorem}\label{thm:clos}
$\sm\tds \leql \clos$.
\end{theorem}
\begin{proof}
Instead of minimizing the closeness centrality, we will try to maximize the \emph{farness}, which is the inverse of the closeness centrality, that is, the sum of all distances from $v$ to another vertex. In our construction, we will build a graph where the vertices with biggest farness correspond to sets in $\C$, and the value of the farness does not depend on the corresponding set, if this set is not disjoint to any other set. If this latter condition is not satisfied, then the farness of the vertex is bigger. In particular, let us consider an instance $I=(X,\C)$ of \sm\tds\, let us assume $X \notin \C$, and let us build a graph in the following way:
\begin{itemize}
\item $V=V_1 \cup V_1' \cup V_2 \cup V_3$, where $V_1$ and $V_1'$ are two disjoint copies of $X$, $V_2=\C$ and $V_3=\{(x,C) \in X \times \C: x \notin C\}$;
\item $V_1 \cup V_1'$ is a clique;
\item for $x \in V_1 \cup V_1'$ and $C \in V_2$, there is an edge from $x$ to $C$ if and only if $x \in C$;
\item for each $(x,C) \in V_3$ and $C' \in V_2$, there is a link between these vertices if and only if $C=C'$.
\end{itemize}

\noindent\textbf{Claim:} the vertex with maximum farness is in $V_3$.
\begin{proof}[Proof of claim]
For each vertex $v \in V_2$, consider a vertex $w \in V_3$ linked to $v$. It is clear that all shortest paths from $w$ to any other vertex must pass through $v$ (which is the only vertex linked to $w$). This means that the farness of $w$ is bigger than the farness of $v$.

For each vertex $v \in V_1 \cup V_1'$, let us consider a vertex $w \in V_2$ linked to $v$. The only vertices which are closer to $w$ than to $v$ are the vertices in $V_3$ attached to $w$, because each other neighbor of $w$ is a neighbor of $v$. These vertices influence the farness of $w$ by $|X|-|C|$, where $C$ is the set corresponding to $w$. However, there are $2(|X|-|C|)$ vertices in $V_1 \cup V_1'$ linked to $v$ and not to $w$ (the elements not in $C$): this proves that the farness of $w$ is bigger than the farness of $v$.
\end{proof}

At this point, let us consider the farness of vertices in $V_3$. In particular, let $(x,C)$ be an element of $V_3$ such that $C \cap C' \neq \emptyset$ for each $C'$: the farness of $(x,C)$ can be exactly computed by considering the classes of vertices in Table \ref{tab:far}.

\begin{table}[h!t]
\caption{The distance from $(x,C)$ to another vertex}\label{tab:far}
\begin{center}

\begin{tabular}{|l|l|r|r|}
\hline
Set & Kind of vertex & Number & distance from $(x,C)$ \\
\hline
$V_1 \cup V_1'$ & vertex in $C$ & $2|C|$ & 2 \\
$V_1 \cup V_1'$ & vertex outside $C$ & $2(|X|-|C|)$ & 3 \\
$V_2$ & $C$ & 1 & 1 \\
$V_2$ & $C' \neq C$ & $|\C|-1$ & 3 \\
$V_3$ & $(x',C)$ & $|X|-|C|$ & 2 \\
$V_3$ & $(x',C')$,$C' \neq C$ & $\sum_{C' \neq C}|X|-|C'|$ & 4 \\
\hline
\end{tabular}
\end{center}
\end{table}

Before computing the farness of $(x,\C)$, we compute $\sum_{C' \neq C}|X|-|C'|=(|\C|-1)|X|-\sum_{C' \neq C}|C'|=(|\C|-1)|X|-\sum_{C' \in \C} |C'|+|C|$.
The farness of $(x,\C)$ is then:

$$4|C|+6(|X|-|C|)+1+3(|\C|-1)+2(|X|-|C|)+4\left((|\C|-1)|X|-\sum_{C' \in \C} |C'|+|C|\right)=$$
$$=4|C|+8|X|-8|C|+1+3|\C|-3+4(|\C|-1)|X|-4\sum_{C' \in \C} |C'|+4|C|=$$
$$=4|\C||X|-4\left(\sum_{C' \in \C} |C'|\right)+3|\C|+4|X|-2.$$

Note that this value does not depend on the particular $(x,C)$ chosen (this was the main goal of our construction). It is clear that if $C \cap C'=\emptyset$, then the farness of each vertex $(x,C)$ and $(x,C')$ is bigger than the value previously computed.

As a consequence, there are two disjoint sets if and only if in the whole graph there is a vertex with farness bigger than $4|\C||X|-4(\sum_{C' \in \C} |C'|)+3|\C|+4|X|-2$, and both this value and the underlying graph can be computed in linear time. 

\end{proof}

\begin{theorem}\label{thm:diam}
$\sm\tds \leql \diam$.
\end{theorem}
\begin{proof}
Given an input $I=(X,\C)$ of $\textsc{\sm TwoDisjointSet}$, construct a split graph $G=(X \cup \C,E)$, where each pair in $X$ is connected, and for each set $C \in \C$ we add an edge from $C$ to its elements.

Since each vertex is at distance 1 from $X$, the diameter is 2 or 3: it is 3 if and only if there exist two different vertices $C,C' \in \C$ with no common neighbor. It is clear that this happens if and only if $C,C'$ are disjoint.
\end{proof}

\begin{theorem}\label{thm:hyper}
$\diamg \leql \hyper$.
\end{theorem}
\begin{proof}
Let $G=(V,E)$ be an input graph for the $\diamg$ problem. The corresponding graph for the $\hyper$ problem is $H=(V',E')$, where $V'=\{x\} \cup V_x\cup \tilde{V},\cup V_y \cup \{y\}$, where $V_x$, $\tilde{V}$ and $V_y$ are disjoint copies of $V$. Edges in $E'$ are defined as follows:
\begin{itemize}
\item $x$ is connected to every vertex in $V_x$ and $y$ is connected to every vertex in $V_y$;
\item corresponding vertices in $V_x$ and $V$ and corresponding vertices in $V$ and $V_y$ are connected;
\item if $(v,w)$ is an edge of $G$, then the copies of $v$ and $w$ in $\tilde{V}$ are linked.
\end{itemize}
In the instance of the $\hyper$ problem we ask if the maximum hyperbolicity of a quadruple containing $x$ and $y$ is bigger than 2. We will prove that this holds if and only if the diameter of $G$ is bigger than 2.

Let us first consider quadruples with vertices $x,y,v,w$ with $v,w \in \tilde{V}$. In these quadruples, $S_1=d(x,y)+d(v,w)=4+d(v,w)$, $S_2=S_3=4$, since the distance from $x$ and $y$ to any vertex in $\tilde{V}$ is 2. The hyperbolicity of such a quadruple is then equal to $d(v,w)$, which reaches the maximum if $v,w$ are a diametral pair. We conclude that if we restrict to such quadruples, then the maximum hyperbolicity equals the diameter.

It remains only to prove that all other quadruples have smaller hyperbolicity. If $v,w \in V_x$ (or $v,w \in V_y$), then $d(v,w)=2$ by passing through $x$, and $S_1=d(x,y)+d(v,w)=6$, $S_2=d(x,v)+d(y,w)=4$ and similarly $S_3=4$. The hyperbolicity is then $2$. Otherwise, if $v \in V_x$ and $w \notin V_x$, $S_2=d(x,w)+d(y,v)=3+d(y,v)\geq 5$ and $S_3=d(x,v)+d(y,w)=1+d(y,w)=1+4-d(y,v)=5-d(y,v)\leq 3$. As a consequence, if vertex $v$ is moved to the corresponding vertex in $\tilde{V}$, $S_3$ can only increase by one and $S_1$ can only decrease by one. Moreover, $S_2$ decreases by one, meaning that, if $S_i'$ are the new values, $S_1'-S_2' \geq S_1-1-(S_2-1)=S_1-S_2$, and the hyperbolicity of the 4-tuple is increased. The proof for all other cases are symmetrical.
\end{proof}

\begin{theorem}\label{thm:matzero}
$\sm\tds \leql \matzero$.
\end{theorem}
\begin{proof}
Let $(X,\C)$ be an instance of the $\sm\tds$ problem. Consider the $|X| \times |\C|$ matrix $M$ containing $1$ in place $(x,C)$ if and only if $x \in \C$. It is clear that $M^TM$ contains a zero in place $C,C'$ if and only if $C \cap C'=\emptyset$. If we want the matrix to be square, it is enough to add $|\C|-|X|$ empty lines after the first $|X|$ lines (this does not change the input size as long as the matrix is stored as an adjacency list).
\end{proof}

\begin{theorem}\label{thm:bgdt}
$\tcov \leql \bgdt$.

$\tcov \leql \bsdp$.
\end{theorem}
\begin{proof}
Let $I=(X,\C)$ be an instance of $\tcov$, and let us create a bipartite graph $G=(V,W,E)$ with $V=X$, $W = \C$ and $E$ defined as the set of pairs $(x,C)$ such that $x \in C$. Then, each 2-covering of $X$ with sets in $\C$ corresponds to a pair of vertices of $W$ that covers $V$.

This way, we proved that $\tcov \leql \bsdp$, by choosing $V$ as the subset to cover.

In order to prove that $\tcov \leql \bgdt$, we add another vertex $v_0 \in V$ connected to any $w \in W$: clearly, if there is a 2-covering of $X$, there is a dominating triple in $G$, which is the 2-covering and $v$.

Viceversa, if there is a dominating triple, one of the vertices of the triple must be in $V$: the other two vertices form a dominating pair (if only one of them is in $W$, the corresponding set is the whole $X$).
\end{proof}

\begin{theorem}[based on Lemma 1 in \cite{Abboud2014}] \label{thm:lalign}
$\sm\tcov \leql \lalign$.
\end{theorem}
\begin{proof}
Let us consider an instance $(X,\C)$ of problem \sm\tcov, assuming $X=0,\dots,k-1$. We code each element $C \in \C$ into a binary string, which is a concatenation of $k$ chunks of the form $01?01$. The $?$ in the $i$-th chunk is $1$ if $i \in C$, $0$ otherwise. We then concatenate all strings for each element $C \in \C$, separated with the string $000$. This way, we constructed the first string.

The binary strings forming the second string are similarly made, but there is a difference in chunks, which are of the form $*1?*1$, where $*$ is the character that may be paired with any other character, and $?$ is $*$ if $i \in C$, $1$ otherwise. these strings are concatenated with $111$ as separator.

We want to prove that two sets $C,C'$ cover $\C$ if and only if the longest common substring of the two strings has length $5k$.

First, let us suppose that there is a common string of length $5k$: we claim first that no separator can occur in this string. This happens because in the first string there is never the sequence $111$ and in the second string there is no $000$. This means that the two equivalent strings of length $5k$ must correspond to two subsets, eventually translated by $-2,-1,1$ or $2$ positions. The sum of the two translation is then smaller than $4$ in absolute value.

However, with a case by case analysis, it is possible to exclude all such translations. For instance, a $-1$ translation would imply that the $6$th element in the first string ($0$) matches the $5$th element in the second string ($1$), and a $+2$ translation would imply that element $5i+2$ (which is $0$ if $i \notin C$) would match element $5i+4$ in the second string, which is $1$.

We concluded that a common substring of length $5k$ must be made by two pieces corresponding to a set. Now it is easy to see that places $5i,5i+1,5i+3,5i+4$ always correspond, while places $5i+2$ correspond if and only if element $i$ is in at least one of the two sets. This concludes the proof.

\end{proof}

\section{Conclusions and Open Problems}\label{sec:conc}

In this paper, we have analyzed many results on quadratic-time problems. We have proved that two of these problems are solvable in truly subquadratic time (recognizing transitive graphs and comparability graphs), and we have provided hardness results for many others. This work can be seen as a starting point to develop many more reductions and include inside this class many new problems.

For instance, it would be really interesting to link these results with all existing reductions on the \tsum\ problem: until now, we have not been able to link it with SETH. However, it is possible to link it with other problems: for example, it is known that the local alignment of strings problem is harder that \tsum\ \cite{Abboud2014}.

Another open problem deals with the radius of graphs: this measure is similar to the diameter, but it looks somehow ``easier'' to compute, for example because the radius of chordal graphs is linear-time computable \cite{ChepoiD94}. The question is whether also this problem can be inserted in our class of quadratic-time hard problems, or a truly subquadratic algorithm exists.

Among other problems that are not in our class, but have no truly subquadratic time algorithm, we find the computation of the transitive reduction of a directed graph (a ``converse'' of the transitive closure), finding maximum flows in networks or finding maximum matchings in bipartite weighted graphs. All these problems are defined in Appendix \ref{app:prob}.

\newpage

\bibliographystyle{plain}
\bibliography{library}

\newpage
\appendix

\section{Problem Definitions}\label{app:prob}

In this section, we will precisely define all the problems we are dealing with.

\subsection*{Hard Problems}

\begin{prob}
{\bet.}{a graph $G=(V,E)$.}{the betweenness centrality of each vertex $v$ of $G$, that is, $$\sum_{v \neq s \neq t \in V} \frac{\text{number of shortest paths from }s\text{ to }t \text{ through }v}{\text{number of shortest paths from }s\text{ to }t}.$$}
\end{prob}

\begin{prob}
{\betv.}{a graph $G=(V,E)$ and a vertex $v \in V$.}{the betweenness centrality of $v$, that is, $$\sum_{v \neq s \neq t \in V} \frac{\text{number of shortest paths from }s\text{ to }t \text{ through }v}{\text{number of shortest paths from }s\text{ to }t}.$$}
\end{prob}

%

\begin{prob}
{\sm\doms.}{a set $X$ and a collection $\C$ of subsets of $X$ such that $|X| < \log^k(|\C|)$ for some $k$.}{\textbf{True} if there are two sets $C,C' \in \C$ such that $C \subseteq C'$, \textbf{False} otherwise.}
\end{prob}

\begin{prob}
{\sm\tcov.}{a set $X$ and a collection $\C$ of subsets of $X$ such that $|X| < \log^k(|\C|)$ for some $k$.}{\textbf{True} if there are two sets $C,C' \in \C$ such that $X=C \cup C'$, \textbf{False} otherwise.}
\end{prob}

\begin{prob}
{\sm\tds.}{a set $X$ and a collection $\C$ of subsets of $X$ such that $|X| < \log^k(|\C|)$ for some $k$.}{\textbf{True} if there are two disjoint sets $C,C' \in \C$, \textbf{False} otherwise.}
\end{prob}

\begin{prob}
{\bgdt.}{a bipartite graph $G=(V,E)$.}{a triple $v,w,x$ such that $V=N(v) \cup N(w) \cup N(x) \cup \{v,w,x\}$, if it exists.}
\end{prob}

\begin{prob}
{\bsdp.}{a graph $G=(V,E)$ and a subset $V'\subseteq V$.}{a pair $v,w$ such that $V'=N(v) \cup N(w) \cup \{v,w\}$, if it exists.}
\end{prob}

\begin{prob}
{\bgdomv.}{a bipartite graph $(V_1,V_2,E)$.}{\textbf{True} if there are vertices $v,w$ such that $N(v) \supseteq N(w)$, \textbf{False} otherwise.}
\end{prob}

\begin{prob}
{\gdt.}{a graph $G=(V,E)$.}{a triple $v,w,x$ such that $V=N(v) \cup N(w) \cup N(x) \cup \{v,w,x\}$, if it exists.}
\end{prob}

\begin{prob}
{\diamg.}{a graph $G$.}{\textbf{True} if $G$ has diameter $2$, \textbf{False} otherwise.}
\end{prob}
\begin{prob}
{\gdomv.}{a graph $(V,E)$.}{\textbf{True} if there are vertices $v,w$ such that $N(v) \supseteq N(w)$, \textbf{False} otherwise.}
\end{prob}


\begin{prob}
{\hyper.}{a graph $G=(V,E)$ and two vertices $x,y$.}{the maximum hyperbolicity of a quadruple $x,y,v,w$, where the hyperbolicity of a quadruple is $S_1-S_2$, where $S_1$ is the maximum sum among $d(x,y)+d(v,w), d(x,v)+d(y,w), d(x,w)+d(y,v)$, and $S_2$ the second maximum sum.}
\end{prob}

\begin{prob}
{\ksatt.}{two sets of variables $\{x_i\}$, $\{y_j\}$ of the same size, a set $C$ of clauses over these variables, such that each clause has at most size $k$, the set of possible evaluations of $x_i$ and the set of possible evaluations of $\{y_j\}$.}{\textbf{True} if there is an evaluation of all variables that satisfies all clauses, \textbf{False} otherwise.}
\end{prob}

\begin{prob}
{\lalign.}{two binary strings with a symbol $*$ that may replace any character.}{The longest common substring of the two strings.}
\end{prob}

\begin{prob}
{\msimp.}{a set $X$ and a collection $\C$ of subsets of $X$.}{the maximum simple family of subsets of $X$ (simple means without inclusion).}
\end{prob}

\begin{prob}
{\clos.}{a graph $G=(V,E)$ and a threshold $\sigma$.}{\textbf{True} if there exists a vertex with closeness centrality smaller than $\sigma$, \textbf{False} otherwise. The closeness centrality of a vertex is defined as $\frac{1}{\sum_{w \in V}d(v,w)}$, where $d(v,w)$ is the distance between vertices $v$ and $w$.}
\end{prob}

\begin{prob}
{\ortbv.}{a collection of binary vectors.}{\textbf{True}, if there are two orthogonal vectors, \textbf{False} otherwise.}
\end{prob}

\begin{prob}
{\doms.}{a set $X$ and a collection $\C$ of subsets of $X$.}{\textbf{True} if there are two sets $C,C' \in \C$ such that $C \subseteq C'$, \textbf{False} otherwise.}
\end{prob}

\begin{prob}
{\diam.}{a split graph $G$.}{\textbf{True} if $G$ has diameter $2$, \textbf{False} otherwise.}
\end{prob}

\begin{prob}
{\hasse.}{a set $X$ and a collection $\C$ of subsets of $X$.}{a tree ordering on the sets in $\C$ by inclusion.}
\end{prob}

\begin{prob}
{\tcov.}{a set $X$ and a collection $\C$ of subsets of $X$.}{\textbf{True} if there are two sets $C,C' \in \C$ such that $X=C \cup C'$, \textbf{False} otherwise.}
\end{prob}

\begin{prob}
{\tds.}{a set $X$ and a collection $\C$ of subsets of $X$.}{\textbf{True} if there are two disjoint sets $C,C' \in \C$, \textbf{False} otherwise.}
\end{prob}

\begin{prob}
{\matzero.}{two $(0-1)$-matrices $M,M'$ implemented as adjacency lists.}{\textbf{True} if $MM'$ contains a $0$, \textbf{False} otherwise.}
\end{prob}

\subsection*{Easy Problems}
\begin{prob}
{ComparabilityRecognition.}{an undirected graph $G$.}{\textbf{True} if there is a transitive orientation of $G$, \textbf{False} otherwise.}
\end{prob}

\begin{prob}
{TransitiveClosure.}{a directed acyclic graph $G$.}{the transitive closure of $G$, that is, the minimum subgraph of $G$ such that if there is a path from a vertex $v$ to a vertex $w$, $E(v,w)$ holds.}
\end{prob}

\subsection{Problems not Classified Yet}

\begin{prob}
{NetworkFlow.}{a graph $G=(V,E)$ and two vertices $v,w$.}{the maximum flow allowed by the network from $v$ to $w$.}
\end{prob}

\begin{prob}
{Radius.}{a graph $G=(V,E)$.}{the radius of $G$, that is, $\min_{v \in V} \max_{w \in V} d(v,w)$.}
\end{prob}

\begin{prob}
{TransitiveReduction.}{a directed acyclic graph $G$.}{the transitive reduction of $G$, that is, the minimum subgraph of $G$ having the same transitive closure.}
\end{prob}

\begin{prob}
{WeightedBiMaximumMatching.}{a bipartite weighted graph $G$.}{a maximum matching of $G$.}
\end{prob}

\end{document}